\documentclass[conference]{IEEEtran}
\IEEEoverridecommandlockouts
\usepackage{cite}
\usepackage{amsmath,amssymb,amsfonts,amsthm}
\usepackage{algorithm,algorithmic}
\usepackage{graphicx}
\usepackage{subfig} 
\usepackage{textcomp}
\usepackage{xcolor}
\usepackage{siunitx}
\usepackage{multicol}

\newtheorem{theorem}{Theorem}
\newtheorem{lemma}{Lemma}

\usepackage[capitalize]{cleveref}
\begin{document}

\title{Age of Information-Oriented Probabilistic Link Scheduling for Device-to-Device Networks
}

\author{
	\IEEEauthorblockN{Lixin Wang$^{1}$, Qian Wang$^{2}$, He (Henry) Chen$^{3}$, Shidong Zhou$^{1}$}
	\IEEEauthorblockA{$^{1}$Department of Electronic Engineering, Tsinghua University, Beijing, China}
 \IEEEauthorblockA{$^{2}$School of Communication Engineering, Hangzhou Dianzi University, Zhejiang, China}

\IEEEauthorblockA{$^{3}$Department of Information Engineering, The Chinese University of Hong Kong, Hong Kong SAR, China} 
	\IEEEauthorblockA{Email: wanglx19@mails.tsinghua.edu.cn, qian.wang@hdu.edu.cn, 
he.chen@ie.cuhk.edu.hk, zhousd@tsinghua.edu.cn}

\thanks{The work of L. Wang and S. Zhou is supported in part by National Natural Science Foundation of China under Grants (62394294, 62394290) and the Fundamental Research Funds for the Central Universities under Grant 2242022k60006. The work of H. Chen is supported in part by RGC General Research Funds (GRF) under Project 14205020, the CUHK direct grant for research under Project 4055166.}
}

\maketitle

\begin{abstract}
This paper focuses on optimizing the long-term average age of information (AoI) in device-to-device (D2D) networks through age-aware link scheduling. The problem is naturally formulated as a Markov decision process (MDP). However, finding the optimal policy for the formulated MDP in its original form is challenging due to the intertwined AoI dynamics of all D2D links. To address this, we propose an age-aware stationary randomized policy that determines the probability of scheduling each link in each time slot based on the AoI of all links and the statistical channel state information among all transceivers. By employing the Lyapunov optimization framework, our policy aims to minimize the Lyapunov drift in every time slot. Nonetheless, this per-slot minimization problem is nonconvex due to cross-link interference in D2D networks, posing significant challenges for real-time decision-making. After analyzing the permutation equivariance property of the optimal solutions to the per-slot problem, we apply a message passing neural network (MPNN), a type of graph neural network that also exhibits permutation equivariance, to optimize the per-slot problem in an unsupervised learning manner. Simulation results demonstrate the superior performance of the proposed age-aware stationary randomized policy over baselines and validate the scalability of our method.
\end{abstract}


\section{Introduction}
\label{section1}





With the rapid advancement of the Internet of Things (IoT), a multitude of time-sensitive applications, including autonomous driving, monitoring systems, and industrial control, have emerged. Compared with conventional metrics like delay, these applications prioritize the timeliness of information from the perspective of the receiving end, which can be effectively characterized by the novel metric known as age of information (AoI)~\cite{kaul2012real}. On the other hand, device-to-device (D2D) communication, a pivotal technology for supporting these applications in IoT networks, can significantly increase spectral efficiency and network capacity by enabling devices to communicate directly with each other without routing through a base station~\cite{tehrani2014device}. 

However, the increasing number of mobile devices leads to a scarcity of frequency resources for D2D communication, making it essential to adopt frequency reuse mechanisms. This means that different D2D links may occupy the same frequency band to transmit simultaneously, leading to cross-link interference. Consequently, it is crucial to properly schedule D2D links and investigate the impact of potential interference between different D2D links on the AoI metric. Previous works have addressed the link scheduling problems in D2D networks using various metrics, such as sum-rate \cite{shen2017fplinq}, throughput \cite{baccelli2013adaptive}, and fairness \cite{baccelli2013adaptive,baccelli2014analysis}. However, these metrics differ from AoI, which is inherently Markovian and often necessitates consideration of long-term average performance. This calls for the design of new link scheduling strategies with a focus on optimizing AoI in D2D networks.

Some existing works have put some effort into investigating AoI-oriented link scheduling problems \cite{kadota2018scheduling, kadota2021minimizing, he2016optimizing, talak2019optimizing}. However, they often fail to consider a practical physical interference model. Specifically, a portion of these works focus on the collision interference model, where only one device is scheduled to access the shared channel at a time, and multiple simultaneous accesses result in failure for all scheduled devices \cite{kadota2018scheduling, kadota2021minimizing}. Given the spatial distribution of different devices in D2D networks, such an oversimplified setup can miss many transmission opportunities, leading to increased AoI. The authors in \cite{he2016optimizing, talak2019optimizing} considered abstracting the \textit{feasible activation set} from a general interference model and focused on scheduling devices within that pre-established set. However, this feasible activation set is assumed to be fixed and may not align with the {time-varying channel} common in real-world scenarios. 

This paper focuses on optimizing the long-term average AoI in D2D networks with practical {Rayleigh fading channels}. The problem presents two main challenges: 1) The optimization objective must be evaluated over the long term, meaning we cannot focus solely on current scheduling decisions but must also consider their future impact on AoI; 2) Scheduling one D2D link causes interference with other links. This cross-link interference is time-varying and must be carefully managed. For example, in a D2D network with two links, significant mutual interference usually prevents simultaneous scheduling. However, channel variations occasionally allow both links to be scheduled together. Effectively exploiting these changing transmission opportunities to minimize AoI has not been thoroughly investigated. A common approach to addressing this issue is to model it as a Markov Decision Process (MDP) in order to derive an optimal decision policy. However, the high-dimensional nature of the considered problem renders standard value iteration methods impractical, making deep reinforcement learning (DRL) a viable alternative \cite{sutton2018reinforcement}. For instance, the authors in \cite{leng2019age} formulated a Markov game and developed a multi-agent DRL algorithm to obtain a power control policy for minimizing the expected average AoI in a wireless ad-hoc network. Similarly, the authors in \cite{liu2022graph} proposed a graph neural network (GNN)-based DRL algorithm to minimize the average AoI by optimizing power control policy in D2D networks. Nonetheless, when a large number of devices are involved, DRL methods often suffer from sample inefficiency, requiring extensive interactions and imposing a heavy burden on real-world applications, not to mention the potential for slow convergence and poor performance.

Given the challenges associated with MDP-based approaches for the considered link scheduling problem, an alternative strategy to reduce the long-term average AoI is to narrow the search space of the policy. The authors in \cite{liu2021age, jones2022minimizing} focused on a class of stationary randomized policies, where each device is scheduled based on a time-invariant probability. These policies can be efficiently obtained through convex optimization using statistical channel state information (CSI) and can be deployed in a distributed manner. However, such stationary randomized policies are \textit{age-independent} because they do not leverage instantaneous system AoI, potentially sacrificing performance. Therefore, we are motivated to incorporate AoI into our scheduling policy. The rationale behind this approach can be intuitively illustrated with a simple 2-link scenario: suppose there are two links in a D2D network that cannot be scheduled simultaneously due to significant interference with each other, and scheduling either one causes the same level of interference on other links. In such a scenario, when only one of these two links can be scheduled, prioritizing the link with the larger AoI will result in a lower average AoI. Based on this insight, this paper aims to develop an \textit{age-aware} stationary randomized policy, where each link is scheduled with a probability determined by the instantaneous system AoI and statistical CSI. {To the best of our knowledge, this is the first work to develop an age-aware link scheduling policy for D2D networks using practical interference models.}

The main contributions of this paper are two-fold: (1) By applying the Lyapunov optimization technique \cite{neely2022stochastic}, we transform the long-term average AoI minimization problem into a per-slot Lyapunov drift minimization problem. We derive an explicit expression for the Lyapunov drift of the age-aware stationary randomized policy. (2) The per-slot Lyapunov drift minimization problem is non-convex, making it inefficient to solve with traditional optimization methods for real-time decision-making. {To circumvent this issue, we propose a {\textit{learn-to-optimize}} approach, which involves mapping the critical parameters of the problem directly to the solution using an artificial neural network.} To that end, we carefully analyze the permutation equivariance (PE) property of the optimal solutions to the per-slot problem and propose applying a message passing neural network (MPNN) \cite{gilmer2017neural}, a popular framework within Graph Neural Networks (GNNs), to directly learn the targeted policy. The MPNN is trained in an unsupervised manner and can produce the scheduling probability for each link based on the system AoI and statistical CSI. We conduct computer simulations to evaluate the age-aware stationary randomized policy generated by the MPNN in various D2D networks, demonstrating its superior performance compared to baseline approaches. Notably, the policy maintains robust performance even when generalized to large-scale D2D networks.

The remainder of this paper is organized as follows. Sec.~\ref{section2} describes the system model, the definition of age-aware stationary randomized policy, and problem formulation. Sec.~\ref{section3} derives the Lyapunov drift for the age-aware policy and elaborates on how to optimize it via MPNN. Sec.~\ref{section4} presents simulation results, and  Sec.~\ref{section5} concludes the paper.

\section{System Model and Problem Formulation}
\label{section2}

\subsection{System Model}

\begin{figure}[tbp]
\centering
\includegraphics[width =0.95 \linewidth]{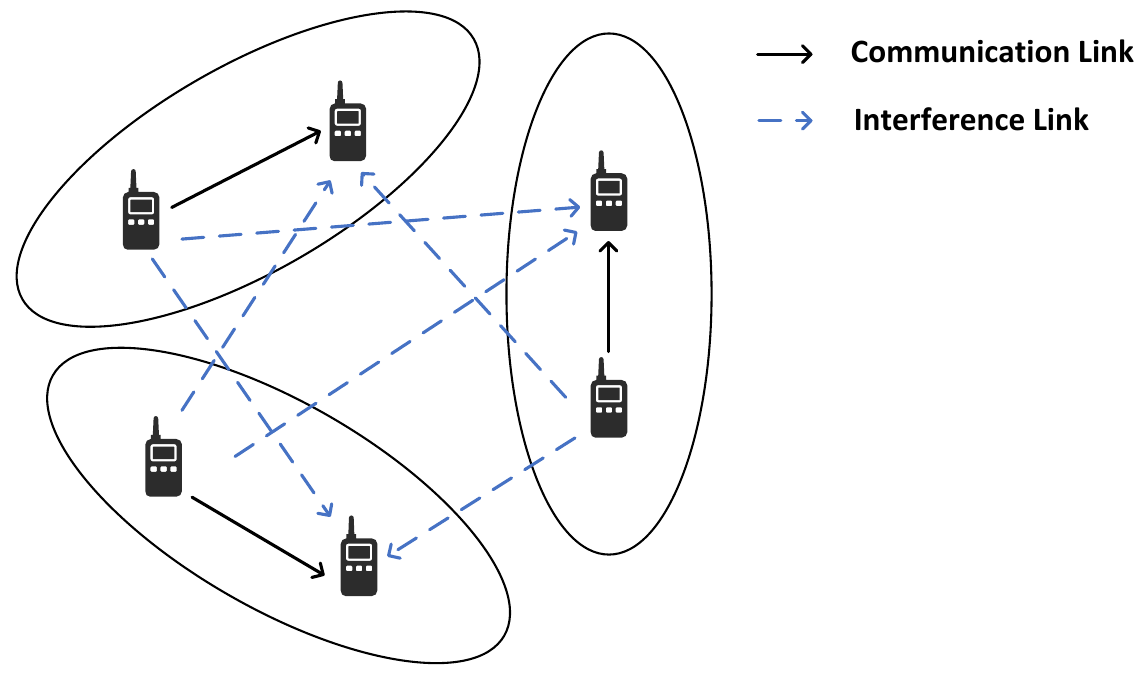}
    \caption{A D2D  communication network with $M = 3$ transmitter-receiver pairs.}
    \label{fig:systemmodel}
    \vspace{-1.5em}
\end{figure}

We consider a D2D wireless network where $M$ independently and randomly distributed transmitter-receiver pairs can simultaneously transmit their respective status updates, as illustrated in Fig.~\ref{fig:systemmodel}. Each transmitter-receiver pair forms a D2D communication link. We use $\mathcal{M} \triangleq \left\{1,2,...,M\right\}$ to denote the index set of D2D links, transmitters and receivers, i.e., link $i$ contains transmitter $i$ and receiver $i$.
Each transmitter and receiver is equipped with a single antenna.
The considered D2D network operates in a time-slotted manner, where the time synchronization among all D2D links could be achieved through the infrastructure (e.g., a base station). At the beginning of every slot, each transmitter can be scheduled to generate a new status update and transmit it to the receiver using the same frequency band. The entire status update process completes within one time slot. We assume that a centralized controller can be employed to schedule the status updates of D2D links, thereby alleviating the negative impact of excessive interference on successful status updates. Similar assumptions have been considered in \cite{yu2011resource, phunchongharn2013resource, liu2022graph}.

The channel gain between transmitter $i$ and receiver $j$ in slot $t$ is denoted by ${h}_{ij}(t) = \sqrt{{h}_{ij}^{l}}{h}_{ij}^{s}(t)$, where ${h}_{ij}^{l}$ denotes the time-invariant large-scale power gain, ${h}_{ij}^{s}(t)$ represents small-scale Rayleigh
fading channel gain, which follows a complex Gaussian distribution $\mathcal{C N}\left(0,1\right)$.
We assume the instantaneous channel state ${h}_{ij}(t)$ remains unchanged during each time slot but varies independently and identically distributed (\textit{i.i.d.}) across different time slots. Specifically, ${h}_{ii}(t)$ is the channel gain of the direct D2D communication link of the $i$-th transmitter-receiver pair, while $h_{ij}(t)$ represents the channel gain of the interference link between transmitter $i$ and receiver $j$, $j  \neq i$. Denote the large-scale power gain matrix by $\mathbf{H}^{l}$, where $\mathbf{H}^{l}_{(i,j)} = h_{ij}^l$. Hereafter, we refer to $\mathbf{H}^{l}$ as the statistical CSI, which is available at the centralized controller.

The link scheduling action in time slot $t$ is denoted by $\boldsymbol{a}(t) = [a_1(t), a_2(t), \ldots, a_M(t)]^T$, where $a_i(t)$ is a binary indicator which is equal to $1$ when transmitter $i$ is scheduled, and $0$ otherwise. Upon performing scheduling action $\boldsymbol{a}(t)$, the signal-to-interference-and-noise ratio (SINR) for receiver $i$ in the $t$-th time slot can be expressed as 
\begin{equation}
\label{SINR_def}
\xi_i(t)=\frac{P_{tx}\left|{h}_{ii}(t) \right|^2a_i(t)}{\sum_{j \neq i} P_{tx}\left|{h}_{ji} (t)\right|^2{a}_j(t)+\sigma^2},
\end{equation}
where $P_{tx}$ is the identical transmission power of every transmitter, $\sigma^2$ is the noise power at each receiver. 

We consider a widely adopted transmission model where a status update is successfully delivered only if the SINR of the receiver is not smaller than a pre-defined threshold $\beta$~\cite{baccelli2013adaptive,he2016optimizing,leng2019age,liu2021age,jones2022minimizing}. Denote by $d_i(t)$ the transmission result of device $i$ in slot $t$. We thus have $d_i(t)$ equals to 1 if $\xi_i(t) \ge \beta$, and 0 otherwise.
Given the transmission model, we are now ready to define the AoI of receiver $i$ as $g_i(t)$. Generally, AoI measures the time elapsed since the latest update received by the receiver was generated~\cite{kaul2012real}. 
Assuming that all scheduled devices can sample and generate fresh information packets for transmission, known as the \textit{generate-at-will} model~\cite{sun2017update},
the AoI $g_i(t+1)$ of receiver $i$ at the beginning of the $(t+1)$-th slot evolves as follows:
\begin{equation}
\label{equ_evolving_AoI}
g_i(t+1)= g_i(t) + 1 - g_i(t)d_i(t).
\end{equation}

\addtolength{\topmargin}{+0.03in}

\subsection{Problem Formulation}
Define the vector $\boldsymbol{g}(t) = [g_1(t), g_2(t),..., g_M(t)]^T$ to represent all receivers' AoI at the beginning of time slot $t$. Denote $\Pi$ as the set of non-anticipated policies, in which link scheduling action $\boldsymbol{a}(t)$ is made based on the scheduling decision history $\left\{\boldsymbol{a}(k)\right\}_{k=1}^{t-1}$, the evolution of AoI $\left\{\boldsymbol{g}(k)\right\}_{k=1}^{t}$, and statistical CSI $\mathbf{H}^{l}$.
{We use statistical CSI for decision-making instead of instantaneous CSI, because instantaneous CSI requires channel estimation for each link in each time slot, which incurs a large signaling overhead. In contrast, statistical CSI is easier to obtain and provides important information about the communication quality of different links.}
The objective is then to find the optimal scheduling policy $\pi \in \Pi$ such that the following long-term average AoI is minimized,
\begin{equation}
\label{problem_def}
 \min _{\pi \in \Pi}J(\pi), 
\end{equation}
where $ J(\pi) =  \lim _{T \rightarrow \infty} \frac{1}{T M} \mathbb{E}_{\pi}\left[\sum_{t=1}^T \sum_{i=1}^M g_i(t)\right]$, and the expectation is with respect to the random channel states and possible random operations in policy $\pi$. 

While problem~\eqref{problem_def} can be formulated as an MDP, searching for the optimal policy within the set $\Pi$ using conventional MDP-based methods is challenging due to the curse of dimensionality. To address this issue, we narrow the policy search space by focusing on the age-aware stationary randomized policy $\pi_A$. For convenience, we may also refer to this policy as age-aware policy in the following sections. 

Under the policy $\pi_A$, each transmitter $i$ will be scheduled independently with probability $p_i(t)$ at any time slot $t$. Define the probability vector $\boldsymbol{p}(t) = [p_1(t),p_2(t),...,p_M(t)]^T$, which is determined by the current AoI vector $\boldsymbol{g}(t)$ and statistical CSI $\mathbf{H}^{l}$. Allowing for some abuse of notation, we can also denote the mapping between $\boldsymbol{g}(t)$, $\mathbf{H}^{l}$ and $\boldsymbol{p}(t)$ as $\pi_{A}$, i.e., $\boldsymbol{p}(t) = \pi_A(\boldsymbol{g}(t),\mathbf{H}^{l})$. This mapping will not change over time, i.e., the policy is stationary.
Let $\Pi_{A}$ denote the set of all such age-aware policies. With this notation, we can reformulate the problem~\eqref{problem_def} as
follows: 
\begin{equation}
\label{problem_def_2}
 \min _{\pi_A \in \Pi_{A}}J(\pi_A).
\end{equation}
We remark that unlike age-independent stationary randomized policy $\pi_S$ developed in~\cite{liu2021age,jones2022minimizing}, the dynamic nature of AoI in our problem results in a time-varying probability vector. This complexity makes it challenging to derive a closed-form solution for $J(\pi_{A})$ and complicates the task of directly obtaining the optimal age-aware policy $\pi_A^\star$. In the following section, we first transform problem~\eqref{problem_def_2} by applying the Lyapunov optimization framework and then elaborate on how to find a well-performing and computationally efficient age-aware policy.

\section{Age-aware policy based on Lyapunov optimization and MPNN}

In this section, we first derive an explicit expression for the Lyapunov drift of the age-aware stationary randomized policy, and then prove that the solution to the Lyapunov drift minimization problem possesses the property of permutation equivariance. Finally, we introduce a learn-to-optimize method based on message passing neural network for real-time decision-making.

\label{section3}

\subsection{Lyapunov Drift for Age-aware Policy}
To proceed, we transform the long-term optimization problem~\eqref{problem_def_2} into a per-slot optimization problem using Lyapunov optimization~\cite{neely2022stochastic}. This approach is inspired by the Max-Weight policy proposed in \cite{kadota2018scheduling, kadota2021minimizing}, which minimizes per-slot Lyapunov drift to achieve near-optimal AoI performance. By designing a Lyapunov function and reducing its expected increase in each time slot, we can maintain the system in a desirable state, where each D2D pair has a relatively low AoI. Before demonstrating the problem transformation, we introduce some key definitions, including network state, Lyapunov function, and Lyapunov drift.

Denote by $S(t) = (\boldsymbol{g}(t), \mathbf{H}^l)$ the network state at the beginning of slot $t$ and define the Lyapunov function as  $L\left(S(t)\right):=\frac{1}{2} \sum_{i=1}^M g_i^2(t)$.
The Lyapunov drift, defined as the expected growth of Lyapunov function after one slot, is given by
\begin{equation}
\label{Drift}
\Delta\left(S(t)\right):=\mathbb{E}_{\pi_A}\left\{L\left(S(t+1)\right)-L\left(S(t)\right) \mid S(t)\right\},
\end{equation}
where the expectation is taken over channel randomness and the random operations in policy $\pi_{A}$.
We present the explicit expression of Lyapunov drift in the following theorem.

\begin{theorem}
\label{lemma1}
    At the beginning of time slot $t$, Lyapunov drift can be expressed as follows:
  \begin{equation}
     \Delta\left(S(t)\right)\! =\!-\sum_{i=1}^M W_i(t)\rho_i p_i{(t)} \prod_{j \neq i}\left(1\!-\!\frac{p_j{(t)}}{1+D_{j i}}\right)\! +\! B(t),
  \end{equation}
where $\rho_i=\exp \left(- {\beta \sigma^2}/({P_{t x}h_{ii}^l}) \right)$, $D_{ji} = {h_{ii}^l}/({\beta h_{ji}^l})$, $W_i(t)=\frac{1}{2} g_i(t)(g_i(t)+2)$, $B(t)=\sum_{i=1}^M(g_i(t)+\frac{1}{2})$, and $p_i{(t)}$ is the $i$-th entry of probability vector $\boldsymbol{p}{(t)} = \pi_{A}(\boldsymbol{g}(t),\mathbf{H}^{l})$.  
\end{theorem}
\begin{proof}See Appendix \ref{appendixA}.\end{proof}

With Theorem~\ref{lemma1}, we now need to optimize the following problem in every slot $t$ to minimize Lyapunov drift in~\eqref{Drift}:
\begin{equation}
\label{problem1_def}
\begin{aligned}
&\min_{\boldsymbol{p}}\!f(\boldsymbol{p}, \boldsymbol{w},\mathbf{H}^{l} )\! =\!-\sum_{i=1}^M \!W_i(t) \rho_i p_i\! \prod_{j \neq i}\!\left(1-\frac{p_j}{1+D_{j i}}\right)\!, \\
&\text { s.t. }\boldsymbol{p}= [p_1,...,p_M]^{T},\boldsymbol{w}= [W_1(t),...,W_M(t)]^{T}, \\
&\quad \quad0 \leq p_i{} \leq 1,\forall i \in \mathcal{M} . 
\end{aligned}
\end{equation}
Let $\hat{\pi}_A^\star$ denote the age-aware policy minimizing~\eqref{problem1_def} in every slot. 
To obtain $\hat{\pi}_{A}^\star$, it is essential to solve problem~\eqref{problem1_def} in every slot. However, problem~\eqref{problem1_def} can be non-convex for our interference model, posing significant challenges in achieving a global optimum. As such, traditional optimization methods, such as numerical optimization and iterative algorithms, may fail to meet the real-time decision-making requirements due to their high computational complexity. We thus propose adopting a {learn-to-optimize} approach, where the critical parameters $\{\boldsymbol{w}, \mathbf{H}^l\}$ of the optimization problem are mapped to the solution $\boldsymbol{p}$ using an artificial neural network. 

Numerous studies have shown that leveraging the structure of the optimization problem's solution to design the corresponding neural network can lead to improved optimization results~\cite{eisen2020optimal,shen2020graph,guo2021learning}. For example, in~\cite{shen2020graph}, the authors discovered that the solution of a weighted sum rate maximization problem exhibits the permutation
equivariance (PE) property. As a result, they were able to utilize message passing neural network (MPNN), a popular GNN framework which also exhibits the PE property, to solve the problem with both high accuracy and scalability. Inspired by this line of research, we investigate whether the solution to problem~\eqref{problem1_def} possesses the PE property and present the following theorem.

\begin{theorem}
\label{PE}
   Considering a D2D network $\mathcal{D}$, problem~\eqref{problem1_def} at the beginning of time slot $t$ can be formulated with critical parameters $\{\boldsymbol{w}, \mathbf{H}^l\}$, and the set of all the optimal solutions to the corresponding problem is denoted by $\Upsilon$. Then, if some permutation operations are applied to the D2D network $\mathcal{D}$, i.e., the reordering of the index of each D2D pair, the corresponding permuted D2D network is denoted by $\mathcal{D}^\prime$ and the reformulated problem~\eqref{problem1_def} at slot $t$ will have the critical parameters $\{\mathbf{P}\boldsymbol{w}, \mathbf{P}\mathbf{H}^l\mathbf{P}^{T}\}$, where permutation matrix $\mathbf{P} \in\{0,1\}^{M \times M}$ satisfies $\mathbf{P}\mathbf{1}=\mathbf{1}$ and $\mathbf{P}^T\mathbf{1}=\mathbf{1}$. Denote by $\Upsilon^{\prime}$ the set of all optimal solutions to the reformulated problem~\eqref{problem1_def}, $\Upsilon$ is permutation equivariant with $\Upsilon^\prime$, i.e., 
$$
\Upsilon^\prime = \left\{\mathbf{P}\boldsymbol{\mu}\mid \boldsymbol{\mu}\in \Upsilon\right\}.
$$
\end{theorem}
\begin{proof}See Appendix \ref{appendix}.\end{proof}

The above advantageous property of problem~\eqref{problem1_def} motivates us to use MPNN to solve it. We will detail this in the next subsection.

\subsection{Optimization via MPNN}
\begin{figure}
\centering
\includegraphics[width =1 \linewidth]{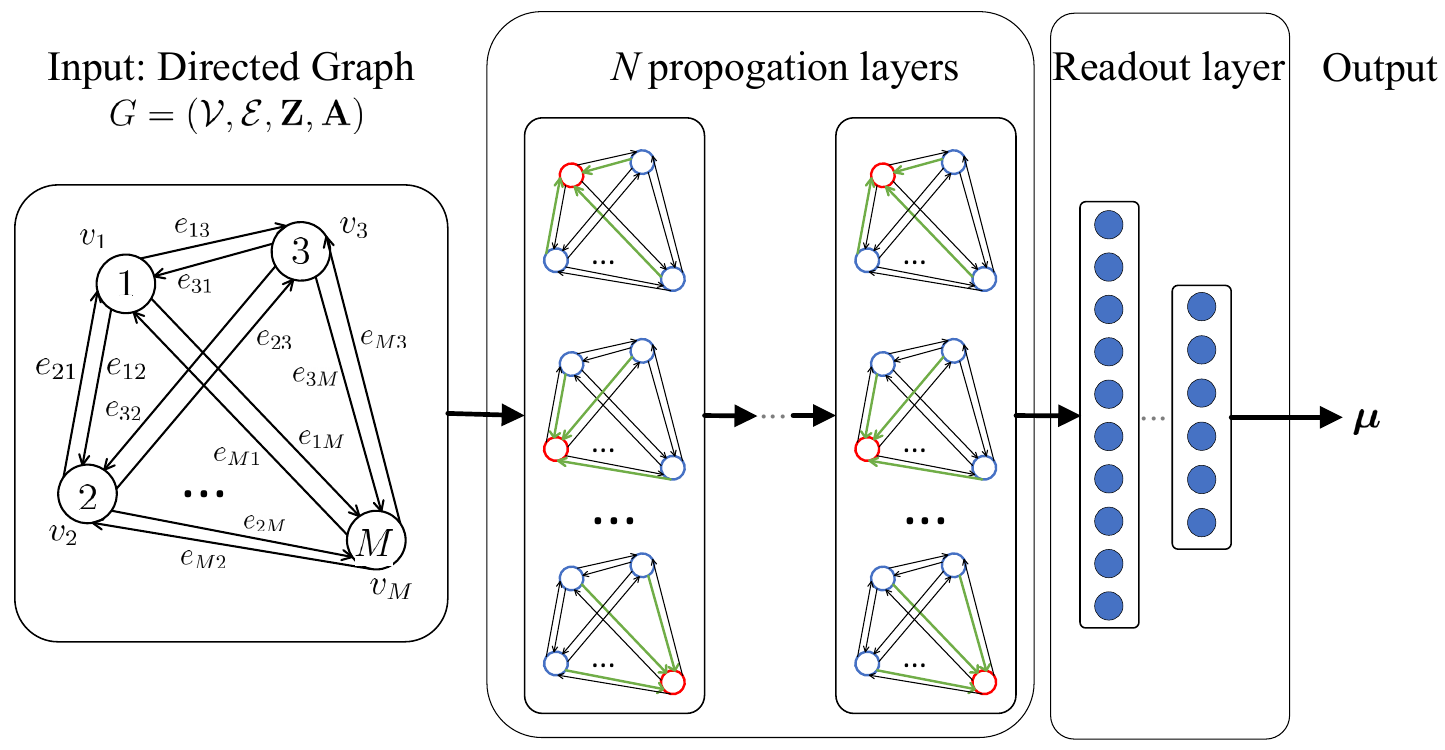}
    \caption{Structure of MPNN for the graph optimization problem.}
    \label{fig:MPNN}
    \vspace{-1em}
\end{figure}
As a type of GNN, MPNN only accepts data input with a graph structure. To proceed, we need to reconsider  problem~\eqref{problem1_def} from the perspective of a directed graph $G$ and transform it into an equivalent graph optimization problem.

\begin{figure}
\centering
\includegraphics[width =1 \linewidth]{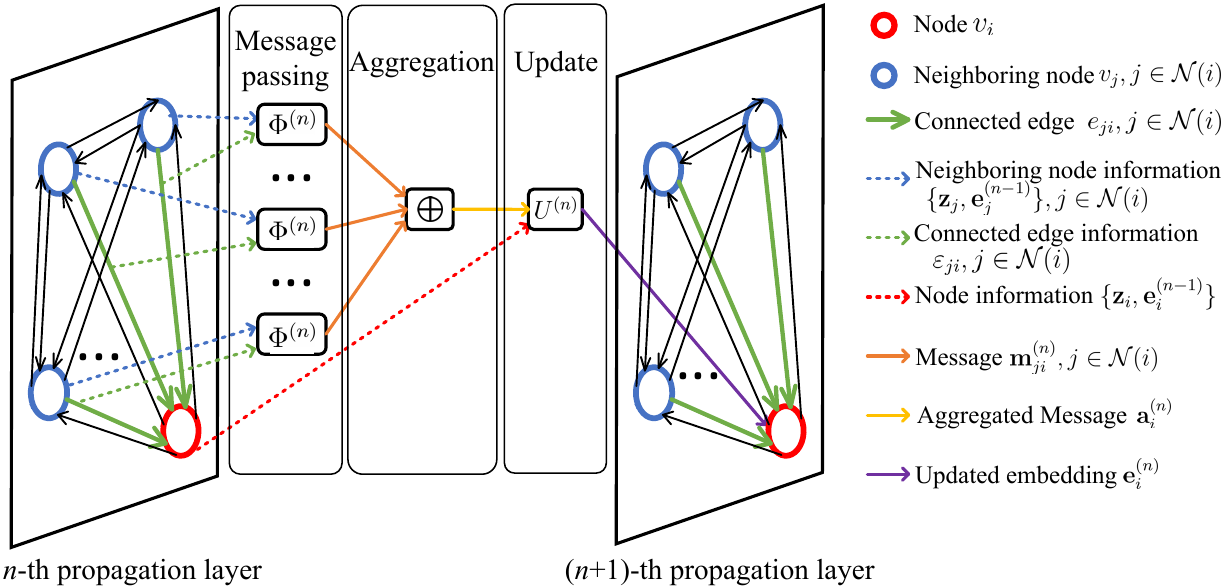}
    \caption{The propagation process of node $v_i$ in the $n$-th propagation layer.}
    \label{fig:MessagePassing}
    \vspace{-2em}
\end{figure}

The $i$-th D2D pair can be viewed as the $i$-th node $v_i$ in the graph $G$. The set of all nodes is denoted by $\mathcal{V}$. The feature of node $v_i$ is defined as $\mathbf{z}_i = [w_i, {h_{ii}^l}]$,  and node feature matrix is defined as $\mathbf{Z} = [\mathbf{z}_1^{T},\mathbf{z}_2^{T}...,\mathbf{z}_M^{T}]^{T} \in \mathbb{R}^{M\times2}$. The weight $w_i = {W_i(t)}/{\max\{W_1(t),W_2(t)...,W_M(t)\}}$, which is a normalization for $W_i(t)$ to facilitate neural network training and inference. The interference link from transmitter $i$ to receiver $j$ can be viewed as the edge $e_{ij}$ from node $i$ to node $j$, and the edge feature $\varepsilon_{ij}$ is ${h_{ij}^l}$. The set of edges is denoted by $\mathcal{E}$. The adjacency feature matrix $\mathbf{A} \in \mathbb{R}^{M\times M}$ can be defined~as
\begin{equation}
\mathbf{A}_{(i,j)}= 
\begin{cases}
0, & \text { if }\{i, j\} \notin \mathcal{E}, \\ 
\varepsilon_{ij}, & \text { otherwise. }
\end{cases}
\end{equation}

Based on all the definitions above, the directed graph associated with our problem can be represented as $G = (\mathcal{V},\mathcal{E},\mathbf{Z},\mathbf{A})$. 
Denote the graph optimization variable by $\boldsymbol{\mu} = [\mu_1, \mu_2, ..., \mu_M]^T \in [0,1]^M$, where $\mu_i$ represents the scheduling probability associated with node $v_i$. Based on the directed graph $G  = (\mathcal{V},\mathcal{E},\mathbf{Z},\mathbf{A})$, the problem~\eqref{problem1_def} in every slot can be reformulated as the following equivalent graph optimization problem: 
\begin{equation}
\label{problem2_def}
\begin{aligned}
& \min _{\boldsymbol{\mu}}   \hat{f}(\boldsymbol{\mu},\mathbf{Z},\mathbf{A})\! = \!-\!\sum_{i=1}^M \mathbf{Z}_{(i,1)} \rho_i \mu_{i} \prod_{j \neq i}\left(1\!-\!\frac{\mu_j}{1+D_{j i}}\right), \\
&\text { s.t. }\rho_i=\exp \left(- \frac{\beta \sigma^2}{P_{t x}\mathbf{Z}_{(i,2)}} \right), \quad D_{ji} = \frac{\mathbf{Z}_{(i,2)}}{\beta \mathbf{A}_{(j,i)} },\\
&\quad \quad \boldsymbol{\mu} = [\mu_1, \mu_2, ..., \mu_M]^T, 0 \leq \mu_i{} \leq 1,\forall i \in \mathcal{M}.\end{aligned}
\end{equation}
As will become clear later, formulating the graph problem~\eqref{problem2_def} will facilitate the training process of the MPNN.

To proceed, we provide a brief introduction of MPNN. An MPNN consists of multiple propagation layers and a readout layer\cite{gilmer2017neural}. The structure of the MPNN is illustrated in Fig.~\ref{fig:MPNN}. Each propagation layer consists of a message passing phase, an aggregation phase, and an update phase. The readout layer may have different forms depending on the specific graph-based task. The primary objective of MPNN is to acquire an informative embedding vector for every node in the graph through the propagation layers and then generate a target value via the readout layer. 
During the propagation phase, a node updates its embedding using its own feature and embedding, and the features of neighboring nodes and connecting edges, and the embeddings of neighboring nodes. The detailed process is depicted in Fig.~\ref{fig:MessagePassing}.
Let $\mathbf{e}_{i}^{(n)} \in \mathbb{R}^E$ denote the embedding representation of node $i$ after the $n$-th propagation layer, where $E$ is the node embedding size. For the $0$-th layer or the input layer, the graph embedding of each node is initialized to zero vector, i.e., $\mathbf{e}_{i}^{(0)}$ = \textbf{0}. Then the update rule of $\mathbf{e}_{i}^{(n)}$ can be described as 
\begin{equation}
\label{MPNN_update}
\begin{aligned}
\text{Message:}&\quad\mathbf{m}_{ji}^{(n)} = \Phi^{(n)}\left(\mathbf{e}_{j}^{(n-1)}, \mathbf{z}_j, \varepsilon_{ji}\right), \forall j \in \mathcal{N}(i),\\ 
\text{Aggregation:}&\quad\mathbf{a}_{i}^{(n)}= \bigoplus_{j \in \mathcal{N}(i)}\mathbf{m}_{ji}^{(n)} ,\\  
\text{Update:}&\quad\mathbf{e}_{i}^{(n)}=U^{(n)}\left(\mathbf{e}_{i}^{(n-1)},\mathbf{z}_i, \mathbf{a}_{i}^{(n)} \right),\nonumber
\end{aligned}
\end{equation}
where $\bigoplus$ denotes a differentiable, permutation invariant aggregation function such as \textit{SUM}, \textit{MIN} or \textit{MAX}, $\mathcal{N}(i)$ denotes the set of neighbors of $i$. The message passing function $\Phi^{(n)}(\cdot)$ and update function $U^{(n)}(\cdot)$ are multi-layer perceptions (MLPs) with parameters $\varphi_n$ and $\psi_n$, respectively. If not specified, the MLPs in this paper all adopt ReLU as the activation function. Aggregation function \textit{MAX} is adopted in this paper for its robustness to feature corruptions\cite{qi2017pointnet}. In MPNN, we will set up $N$ propagation layers and $N$ should not be too large to avoid over-smoothing, i.e., all the nodes will have similar embeddings. An efficient design is that we let the message passing functions and update functions of $N$ layers share the same MLP parameters $\varphi$ and $\psi$ respectively.
After $N$-layer propagation, the readout layer needs to map the learned embedding $\mathbf{e}_{i}^{(N)}$ to the scheduling probability $\mu_{i}$ of node $i$. Since $\mu_{i} \in [0,1]$, the readout layer can be designed as an MLP with parameters $\Omega$ and a Sigmoid activation function.

The entire MPNN is trained using the unsupervised learning paradigm. It employs the loss function 
$\mathbb{E}\left\{\hat{f}(\boldsymbol{\mu}(\varphi,\psi,\Omega),\mathbf{Z},\mathbf{A})\right\}$ and utilizes stochastic gradient descent to update network parameters $\{\varphi,\psi,\Omega\}$, where $\hat{f}$ is defined in the graph problem~\eqref{problem2_def}. During the training phase, a batch of different training samples $\{\mathbf{Z},\mathbf{A}\}$ is used to compute gradients. During the inference phase, we use statistical CSI and the system's AoI to generate the randomized scheduling policy for each time slot.

\section{Simulation Results}
\label{section4}

In this section, we compare the performance of the proposed age-aware policy based on Lyapunov optimization and MPNN with the following baselines:
\begin{itemize}
    \item \textbf{Optimal Age-independent Stationary Randomized policy}\cite{liu2021age,jones2022minimizing}: Obtain $\pi_{S}^\star$ by minimizing $J(\pi_{S})$ via convex optimization.
    \item \textbf{Optimal Proportional Fairness}\cite{jones2022minimizing}: An age-independent stationary randomized policy $\pi_{S}$ that minimizes the sum log of average AoI.
    \item \textbf{Greedy Policy}: Only schedule the D2D pair with the highest AoI.
\end{itemize}

\begin{figure}[tbp]
\centering
\includegraphics[width =1 \linewidth]{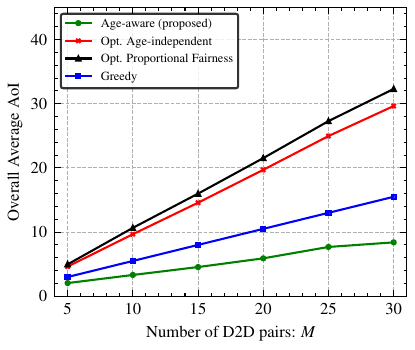}
    \caption{Overall AoI performance comparison for different $M$ over 500 layouts.}
    \label{trained_performance}
    \vspace{-1em}
\end{figure}

\subsection{Parameter Setting}
The D2D network is within an $L \times L $ area, where $L$ is the area length. The link density is set to ${M}/{L^2}$ links/$\text{m}^2$. If not specified, $L = 500\mathrm{m}$ for all the training and testing scenarios. 
All transmitters are randomly distributed in the area and each receiver is uniformly distributed within $2\rm{m}\sim40\rm{m}$ from the corresponding transmitter. 
The D2D communication channel follows ultra high frequency model in ITU-1411~\cite{series2017propagation}.
The large-scale power gain in dB, i.e., path-loss, is captured~by
\begin{equation}
\label{pl}
L_{p l}=L_{b p}+6+\left\{\begin{array}{l}
20 \log _{10}\left(\frac{d}{R_{b p}}\right), \text { for } \quad d \leq R_{b p} \\
40 \log _{10}\left(\frac{d}{R_{b p}}\right), \text { for } \quad d>R_{b p}
\end{array},\right.
\end{equation}
where $d$ is the distance between the transmitter and the receiver, $R_{b p}$ is the breakpoint distance defined as $ R_{b p}= \frac{4 h_1 h_2}{\lambda}$ with $h_1$ and $h_2$ being antenna heights of the transmitter and receiver, and $\lambda$ is the carrier wavelength. $L_{b p}$ in \eqref{pl} is defined as $L_{b p}=\left|20 \log _{10}\left(\frac{\lambda^2}{8 \pi h_1 h_2}\right)\right|$, representing the basic transmission loss at the breakpoint. In our simulation, the used carrier frequency is $2.4$ GHz, with a bandwidth of $5$ MHz. 
The antenna heights of both transmitter and receiver are set to be $1.5$m, and the antenna gain is $2.5$ dBi. The transmission power $P_{tx}$ is $40$ dBm, and the noise power spectral density is $-169$ dBm/Hz. The threshold $\beta = 1023$.
We generate weight $w_i$ from a uniform distribution in $[0, 1]$ for any $i \in \mathcal{M}$. 

\begin{figure}[tbp]
\centering
\includegraphics[width =0.95 \linewidth]{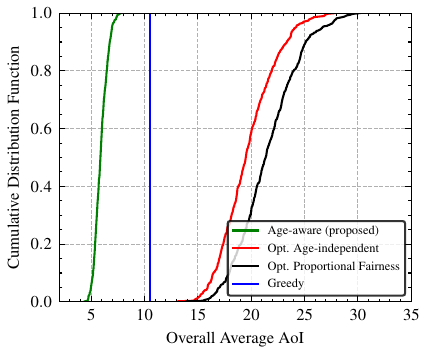}
    \caption{CDF of overall average AoI over 500 layouts: $M = 20$.}
    \label{CDF_20links}
    \vspace{-2em}
\end{figure}

For the parameter setting of MPNN, we set the dimension $E$ of node embedding to $8$, the layer number $N = 3$, the number of hidden neurons in all the message passing functions as $\{11,32,32\}$, the number of hidden neurons in all the update functions as $\{42,16,8\}$, the number of hidden neurons in readout layer as $\{8, 16, 1\}$. We used the Adam optimizer~\cite{kingma2014adam} with a learning rate of $0.002$ and a decay factor of $0.9$. The training process consists of 100 epochs of batch training with batch size $50$. For the input data of MPNN, we note that the interference link is considered only if the transmitter $i$ and receiver $j$ are within $500$ meters. For the data normalization, both the large-scale power gain $h_{ii}^l$ in node feature and the large-scale power gain $h_{ij}^l$ in edge feature are in dB and normalized by z-score $\Tilde{x}= (x - \Bar{x})/\hat{x}$, where $x$ is the input data, $\Bar{x}$ is the average value of data, $\hat{x}$ is the corresponding standard deviation. The weight $w_i$ is normalized by $\Tilde{w}_i = {w_i}/{\max\{w_1,w_2...,w_M\}}$.

\vspace{-1em}

\subsection{Simulation Results}


The long-term average AoI performance of each scheduling policy is evaluated over 20,000 time slots for every given D2D layout. For our proposed age-aware policy, the MPNN for solving problem~\eqref{problem2_def} is trained with 50000 samples of $\{\mathbf{Z},\mathbf{A}\}$ for $M = 5,10,15,20,25,30$ in $500 \text{m}\times 500 \text{m}$ area, each of which contains independently randomly generated D2D layout and weights. In the testing phase, the trained MPNN generates the scheduling probability vector at the beginning of each time slot, using statistical CSI and the current system AoI as inputs.
By averaging 500 testing layouts, Fig.~\ref{trained_performance} depicts the overall AoI performance of the age-aware policy and other baselines versus $M$. It can be observed that the proposed policy outperforms other policies, while the optimal age-independent stationary randomized policy can even exhibit worse performance than the greedy policy. One possible reason is that the optimal age-independent stationary randomized policy may often schedule links with low AoI while ignoring links with very high AoI. This observation validates the importance of incorporating AoI for link scheduling. Moreover, We also present the cumulative distribution function (CDF) of overall AoI performance in Fig.~\ref{CDF_20links} for $M = 20$ over 500 D2D testing layouts. The CDF curve of the age-aware policy always locates in the upper left compared to other policies, demonstrating the stable inference performance of the MPNN.

Next, we evaluate the generalization capability of the proposed method, specifically whether the MPNN trained on D2D networks with tens of D2D pairs can be directly deployed on D2D networks with hundreds of D2D pairs without fine-tuning.  To this end, we train the MPNN in a $500 \text{m}\times 500 \text{m}$ layout with 20 D2D pairs using 50000 training samples and test its scalability in the following two cases: same link density and different link densities. 

For the same link density case, we test the trained MPNN across 100 testing layouts, each containing different numbers of D2D pairs, {with $L$ increasing as $M$ increases}. The results are shown in Fig.~\ref{generalize_same}. Although the number of D2D pairs in the testing stage is different from that in the training stage, the performance of the age-aware policy remains stable and significantly outperforms the baselines.

For the different link density cases, we test the trained MPNN in 100 layouts, all within a $500 \text{m}\times 500\text{m}$ area but with varying numbers of D2D pairs. The results, shown in Fig.~\ref{generalize_diff}, indicate that the age-aware policy maintains a performance advantage over the baselines even when the link density is increased tenfold. These results collectively demonstrate the scalability of MPNN in terms of AoI performance.

\begin{figure}[tbp]
\centering
\includegraphics[width =0.95 \linewidth]{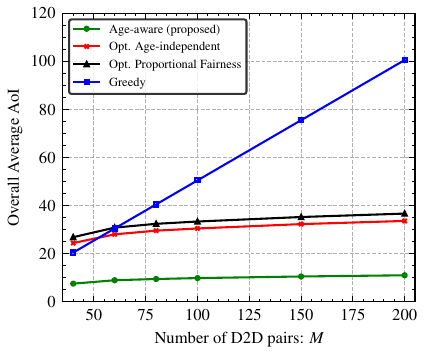}
    \caption{Generalize to the same density.}
    \label{generalize_same}
    \vspace{-1em}
\end{figure}
\begin{figure}[tbp]
\centering
\includegraphics[width =0.95 \linewidth]{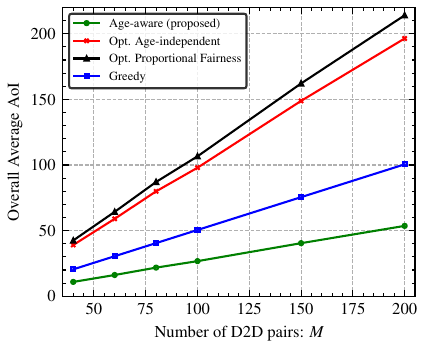}
    \caption{Generalize to different densities.}
    \label{generalize_diff}
    \vspace{-2em}
\end{figure}

\section{Conclusion}
\label{section5}
In this paper, we developed an age-aware link scheduling policy for long-term average AoI minimization in device-to-device (D2D) networks. Specifically, we proposed a message passing neural network (MPNN)-based method to generate the age-aware stationary randomized policy for real-time link scheduling, using only statistical CSI and system age of information (AoI). The MPNN was trained in an unsupervised manner by minimizing the Lyapunov drift of the original long-term average AoI minimization problem.
The performance of the age-aware stationary randomized policy has been demonstrated to be stable and superior to several baselines. Additionally, the proposed method exhibits strong scalability capability.




\bibliographystyle{./IEEEtran}
\bibliography{./IEEEexample}

\begin{thebibliography}{10}
\providecommand{\url}[1]{#1}
\csname url@samestyle\endcsname
\providecommand{\newblock}{\relax}
\providecommand{\bibinfo}[2]{#2}
\providecommand{\BIBentrySTDinterwordspacing}{\spaceskip=0pt\relax}
\providecommand{\BIBentryALTinterwordstretchfactor}{4}
\providecommand{\BIBentryALTinterwordspacing}{\spaceskip=\fontdimen2\font plus
\BIBentryALTinterwordstretchfactor\fontdimen3\font minus \fontdimen4\font\relax}
\providecommand{\BIBforeignlanguage}[2]{{%
\expandafter\ifx\csname l@#1\endcsname\relax
\typeout{** WARNING: IEEEtran.bst: No hyphenation pattern has been}%
\typeout{** loaded for the language `#1'. Using the pattern for}%
\typeout{** the default language instead.}%
\else
\language=\csname l@#1\endcsname
\fi
#2}}
\providecommand{\BIBdecl}{\relax}
\BIBdecl

\bibitem{kaul2012real}
S.~Kaul, R.~Yates, and M.~Gruteser, ``Real-time status: How often should one update?'' in \emph{Proc. IEEE INFOCOM}, Mar. 2012, pp. 2731--2735.

\bibitem{tehrani2014device}
M.~N. Tehrani, M.~Uysal, and H.~Yanikomeroglu, ``Device-to-device communication in {5G} cellular networks: challenges, solutions, and future directions,'' \emph{IEEE Commun. Mag.}, vol.~52, no.~5, pp. 86--92, May 2014.

\bibitem{shen2017fplinq}
K.~Shen and W.~Yu, ``{FPLinQ}: A cooperative spectrum sharing strategy for device-to-device communications,'' in \emph{Proc. IEEE Int. Symp. Inf. Theory (ISIT)}, Jun. 2017, pp. 2323--2327.

\bibitem{baccelli2013adaptive}
F.~Baccelli and C.~Singh, ``Adaptive spatial aloha, fairness and stochastic geometry,'' in \emph{Proc. 11th Int. Symp. Model. Optim. Mobile Ad Hoc Wireless Netw. (WiOpt)}, May 2013, pp. 7--14.

\bibitem{baccelli2014analysis}
F.~Baccelli, B.~B{\l}aszczyszyn, and C.~Singh, ``Analysis of a proportionally fair and locally adaptive spatial aloha in {Poisson} networks,'' in \emph{Proc. IEEE INFOCOM}, Jul. 2014, pp. 2544--2552.

\bibitem{kadota2018scheduling}
I.~Kadota, A.~Sinha, E.~Uysal-Biyikoglu, R.~Singh, and E.~Modiano, ``Scheduling policies for minimizing age of information in broadcast wireless networks,'' \emph{IEEE/ACM Trans. Netw.}, vol.~26, no.~6, pp. 2637--2650, Dec. 2018.

\bibitem{kadota2021minimizing}
I.~Kadota and E.~Modiano, ``Minimizing the age of information in wireless networks with stochastic arrivals,'' \emph{IEEE Trans. Mobile Comput.}, vol.~20, no.~3, pp. 1173--1185, Mar. 2021.

\bibitem{he2016optimizing}
Q.~He, D.~Yuan, and A.~Ephremides, ``Optimizing freshness of information: On minimum age link scheduling in wireless systems,'' in \emph{Proc. 14th Int. Symp. Model. Optim. Mobile Ad Hoc Wireless Netw. (WiOpt)}, Jun. 2016, pp. 1--8.

\bibitem{talak2019optimizing}
R.~Talak, S.~Karaman, and E.~Modiano, ``Optimizing information freshness in wireless networks under general interference constraints,'' \emph{IEEE/ACM Trans. Netw.}, vol.~28, no.~1, pp. 15--28, Dec. 2019.

\bibitem{sutton2018reinforcement}
R.~S. Sutton and A.~G. Barto, \emph{Reinforcement learning: An introduction}.\hskip 1em plus 0.5em minus 0.4em\relax MIT press, 2018.

\bibitem{leng2019age}
S.~Leng and A.~Yener, ``Age of information minimization for wireless ad hoc networks: A deep reinforcement learning approach,'' in \emph{Proc. IEEE Global Commun. Conf. (GLOBECOM)}, Dec. 2019, pp. 1--6.

\bibitem{liu2022graph}
Y.~Liu, C.~She, W.~Hardjawana, and B.~Vucetic, ``Graph neural networks for timely updates of short packets in interference-limited networks,'' in \emph{Proc. Asilomar Conf. Signals Syst. Comput. (ACSSC)}, Oct. 2022, pp. 1050--1054.

\bibitem{liu2021age}
Z.~Liu, Z.~Chen, L.~Luo, M.~Hua, W.~Li, and B.~Xia, ``Age of information-based scheduling for wireless device-to-device communications using deep learning,'' in \emph{Proc. IEEE Wireless Commun. Netw. Conf. (WCNC)}, Mar. 2021, pp. 1--6.

\bibitem{jones2022minimizing}
N.~Jones and E.~Modiano, ``Minimizing age of information in spatially distributed random access wireless networks,'' in \emph{Proc. IEEE INFOCOM}, May 2023, pp. 1--10.

\bibitem{neely2022stochastic}
M.~Neely, \emph{Stochastic network optimization with application to communication and queueing systems}.\hskip 1em plus 0.5em minus 0.4em\relax Springer Nature, 2022.

\bibitem{gilmer2017neural}
J.~Gilmer, S.~S. Schoenholz, P.~F. Riley, O.~Vinyals, and G.~E. Dahl, ``Neural message passing for quantum chemistry,'' in \emph{Proc. 34th Int. Conf. Mach. Learn. (ICML)}, Aug. 2017, pp. 1263--1272.

\bibitem{yu2011resource}
C.-H. Yu, K.~Doppler, C.~B. Ribeiro, and O.~Tirkkonen, ``Resource sharing optimization for device-to-device communication underlaying cellular networks,'' \emph{IEEE Trans. Commun.}, vol.~10, no.~8, pp. 2752--2763, Jun. 2011.

\bibitem{phunchongharn2013resource}
P.~Phunchongharn, E.~Hossain, and D.~I. Kim, ``Resource allocation for device-to-device communications underlaying {LTE}-advanced networks,'' \emph{IEEE Wireless Commun. Mag.}, vol.~20, no.~4, pp. 91--100, Aug. 2013.

\bibitem{sun2017update}
Y.~Sun, E.~Uysal-Biyikoglu, R.~D. Yates, C.~E. Koksal, and N.~B. Shroff, ``Update or wait: How to keep your data fresh,'' \emph{IEEE Trans. Inf. Theory}, vol.~63, no.~11, pp. 7492--7508, Nov. 2017.

\bibitem{eisen2020optimal}
M.~Eisen and A.~Ribeiro, ``Optimal wireless resource allocation with random edge graph neural networks,'' \emph{IEEE Trans. Signal Process.}, vol.~68, pp. 2977--2991, April. 2020.

\bibitem{shen2020graph}
Y.~Shen, Y.~Shi, J.~Zhang, and K.~B. Letaief, ``Graph neural networks for scalable radio resource management: Architecture design and theoretical analysis,'' \emph{IEEE J. Sel. Areas Commun.}, vol.~39, no.~1, pp. 101--115, Jan. 2021.

\bibitem{guo2021learning}
J.~Guo and C.~Yang, ``Learning power allocation for multi-cell-multi-user systems with heterogeneous graph neural networks,'' \emph{IEEE Trans. Wireless Commun.}, vol.~21, no.~2, pp. 884--897, Aug. 2021.

\bibitem{qi2017pointnet}
R.~Q. Charles, H.~Su, M.~Kaichun, and L.~J. Guibas, ``{PointNet}: Deep learning on point sets for 3d classification and segmentation,'' in \emph{Proc. IEEE Conf. Comput. Vis. Pattern Recognit. (CVPR)}, Jul. 2017, pp. 77--85.

\bibitem{series2017propagation}
\emph{Propagation Data and Prediction Methods for the Planning of ShortRange Outdoor Radiocommunication Systems and Radio Local Area Networks in the Frequency Range 300 MHz to 100 GHz}, document I. R. P.1411-11 2021.

\bibitem{kingma2014adam}
D.~P. Kingma and J.~Ba, ``{Adam}: A method for stochastic optimization,'' in \emph{Proc. Int. Conf. Learn. Represent. (ICLR)}, Dec. 2014, pp. 1--15.

\end{thebibliography}

\vspace{-0.5em}
\appendices

\section{Proof of Theorem~\ref{lemma1}}
\vspace{-0.5em}
\label{appendixA}
    First, by the definition of Lyapunov function $L(S(t))$ and incorporating \eqref{equ_evolving_AoI} and \eqref{Drift}, we have
    \begin{equation}
    \label{DeltaPi}       \Delta\left(S(t)\right) = -\sum_{i=1}^M\mathbb{E}_{\pi_{A}} \{d_i(t)\mid S(t)\}W_i(t)+B(t)  
    \end{equation}
    for any age-aware policy $\pi_A \in \Pi_{A}$.

   To proceed, we first deal with the term $\mathbb{E}_{\pi_{A}} \{d_i(t)\mid S(t)\}$, which represents the transmission success probability of the transmitter $i$ in slot $t$ by the definition of $d_i(t)$. Under the scheduling policy $\pi_{A}$, each transmitter is scheduled to generate and transmit updates independently with probability $p_i(t)$, which is the $i$-th entry in the probability vector $\boldsymbol{p}{(t)} = \pi_{A}(\boldsymbol{g}(t),\mathbf{H}^{l})$. Then the transmission success probability of the transmitter $i$ is the product of its scheduling probability and conditional success probability, thus we have
\begin{equation}
\label{pikE}
    \mathbb{E}_{\pi_{A}} \{d_i(t)\mid S(t)\} = p_i(t)\mathbb{P}\left[\xi_i(t)\ge\beta \mid a_i(t) = 1 \right ].
\end{equation}  

The conditional success probability can be obtained using a procedure similar to that in \cite{liu2021age,jones2022minimizing,baccelli2013adaptive}. Specifically, we first condition on the channel randomness of interference links and the randomized transmissions of other transmitters to~have
\begin{equation}
\label{condition2}
\begin{aligned}
\mathbb{P} [\xi_i(t)\ge\beta \mid a_i(t) = &1, h_{ji}(t),a_j(t),\forall j\neq i  ]\\
\overset{(a)}{=} \mathbb{P}\Bigg[ \left|{h}_{ii}(t) \right|^2 \ge &\frac{\beta(\sum_{ j \neq i}P_{tx}\left|{h}_{ji} (t)\right|^2{a}_j(t)+\sigma^2)}{P_{tx}} \\ 
&\Bigg| h_{ji}(t),a_j(t),\forall j\neq i \Bigg ] \\
 \overset{(b)}{=} \exp \Bigg(-\frac{\beta}{h_{ii}^l}\Bigg(&\sum_{j \neq i} \left|{h}_{ji} (t)\right|^2 a_j(t) +\frac{\sigma^2}{P_{tx}}\Bigg)\Bigg)\\
 \overset{(c)}{=} \rho_i \prod_{j\neq i}\exp \Bigg(-&\frac{\left|{h}_{ji}^{s} (t)\right|^2 a_j(t)}{D_{ji}} \Bigg),\\ 
\end{aligned}
\end{equation}  
where equation $(a)$ is derived from the definition of SINR for receiver $i$ in~\eqref{SINR_def}, equation $(b)$ is due to the fact that $\left|h_{ii}(t)\right|^2 = h_{ii}^l \left|h_{ii}^s(t)\right|^2$ and $\left|h_{ii}^s(t)\right|^2 $ follows an exponential distribution {with a mean of 1}, equation $(c)$ is by the notation $\rho_i=\exp \left(- {\beta \sigma^2}/({P_{t x}h_{ii}^l}) \right)$ and $D_{ji} = {h_{ii}^l}/({\beta h_{ji}^l})$. After that, by taking expectation on \eqref{condition2} with respect to the independent randomized transmission of other transmitters, we have 
\begin{equation}
\label{condition1}
\begin{aligned}
&\mathbb{P} [\xi_i(t)\ge\beta \mid a_i(t) = 1, h_{ji}(t),\forall j\neq i  ]\\
 & = \mathbb{E}_{\{a_j(t),\forall j \neq i\}} \Bigg\{\rho_i \prod_{j\neq i}\exp \Bigg(-\frac{\left|{h}_{ji}^{s} (t)\right|^2 a_j(t)}{D_{ji}} \Bigg)\Bigg\}\\ 
  & \overset{(d)}{=} \rho_i\prod_{j\neq i}\mathbb{E}_{\{a_j(t),\forall j \neq i\}} \Bigg\{ \exp \Bigg(-\frac{\left|{h}_{ji}^{s} (t)\right|^2 a_j(t)}{D_{ji}} \Bigg)\Bigg\}\\ 
    & \overset{(e)}{=} \rho_i\prod_{j\neq i} \Bigg [ 1- p_j(t) + p_j(t)\exp \Bigg(-\frac{\left|{h}_{ji}^{s} (t)\right|^2 }{D_{ji}} \Bigg)\Bigg ],\\     
\end{aligned}
\end{equation}  
where equation $(d)$ follows by the independent scheduling between transmitters, equation $(e)$ is obtained because $a_j(t)$ follows Bernoulli distribution ${\rm Ber}(p_j(t))$. Finally, by taking expectation on~\eqref{condition1} with respect to the channel randomness of interference links, we have
\begin{equation}
\label{condition}
\begin{aligned}
&\mathbb{P} [\xi_i(t)\ge\beta \mid a_i(t) = 1  ]\\
  = &\mathbb{E}_{\{h_{ji}(t),\forall j\neq i\}}\Bigg\{ \rho_i\prod_{j\neq i} \Bigg [ 1- p_j(t) + \frac{p_j(t)}{\exp ({\left|{h}_{ji}^{s} (t)\right|^2 }/{D_{ji}} )}\Bigg ]\Bigg\}\\ 
  = &\rho_i \prod_{j \neq i}\left(1-\frac{p_j{(t)}}{1+D_{j i}}\right),
\end{aligned}
\end{equation}  
where the last equation is because small-scale channel gain ${h}_{ji}^{s} (t)$ is independent across interference links and it follows an exponential distribution with a mean of 1.

We complete the proof by substituting \eqref{pikE} and \eqref{condition} into~\eqref{DeltaPi}.

\section{Proof of Theorem~\ref{PE}}
\label{appendix}
Before proving Theorem~\ref{PE}, we first show the permutation invariance of $f$ in problem~\eqref{problem1_def}.
\begin{lemma}[Permutation Invariance of $f$]
\label{PI}
    Consider the  problem~\eqref{problem1_def} with critical parameter $\{\boldsymbol{w},\mathbf{H}^l\}$, and $\mu = [\mu_1, \mu_2, ..., \mu_M]$ is any feasible solution to the problem ~\eqref{problem1_def}. For any permutation matrix $\mathbf{P} \in \{0,1\}^{M \times M}$ satisfies $\mathbf{P}\mathbf{1}=\mathbf{1}$ and $\mathbf{P}^T\mathbf{1}=\mathbf{1}$, the function $f$ in problem~\eqref{problem1_def} is permutation invariant, i.e., 
$$
f(\boldsymbol{\mu},\boldsymbol{w},\mathbf{H}^l )=f\left(\mathbf{P}\boldsymbol{\mu}, \mathbf{P}\boldsymbol{w},\mathbf{P} \mathbf{H}^l\mathbf{P} ^ { T }  \right) .
$$
\end{lemma}
\begin{proof}
    Define $\xi_i = W_i(t) \rho_i \mu_{i} \prod_{j \neq i}\left(1-\frac{\mu_j}{1+D_{j i}}\right)$,  
    and we have $f(\boldsymbol{\mu},\boldsymbol{w},\mathbf{H}^l) = \sum_{i=1}^M \xi_i$. 
    By observing the expression of  $\xi_i$, any permutation to the  variable and critical parameters only leads to a reordering of the sequence $\{\xi_1,\xi_2,...,\xi_M\}$, while the overall summation remains unchanged. This implies that the function $f$ is permutation invariant.
\end{proof}
Based on Lemma~\ref{PI}, we can prove the permutation equivariance of $\Upsilon$. Define a set $\hat{\Upsilon} = \left\{\mathbf{P}\boldsymbol{\mu}\mid \boldsymbol{\mu}\in \Upsilon\right\}$. To prove $\Upsilon^\prime = \hat{\Upsilon}$, it suffices to prove $\hat{\Upsilon} \subseteq  \Upsilon^\prime$ and $\Upsilon^\prime \subseteq  \hat{\Upsilon}$ . We first show $\hat{\Upsilon} \subseteq  \Upsilon^\prime$ holds. Any optimal solution $\boldsymbol{\mu} \in \Upsilon$ achieves the optimum of problem~\eqref{problem1_def} based on critical parameters $\{\boldsymbol{w}, \mathbf{H}^l\}$. The permutation invariance of $f$ promises that solution $\mathbf{P}\boldsymbol{\mu}$ will also achieve the same optimum of the reformulated problem~\eqref{problem1_def} based on the permuted parameters $\{\mathbf{P}\boldsymbol{ w},\mathbf{P} \mathbf{H}^l\mathbf{P} ^ { T }\}$. Therefore, $\mathbf{P}\boldsymbol{\mu}\in \Upsilon^\prime$, which verifies $\hat{\Upsilon} \subseteq  \Upsilon^\prime$. We turn to prove  $\Upsilon^\prime \subseteq  \hat{\Upsilon}$. Specifically, based on the above definitions, D2D network $\mathcal{D}^\prime$ can be transformed into D2D network $\mathcal{D}$ with some inverse permutation operations, and the corresponding critical parameters $\{\mathbf{P}\boldsymbol{w},\mathbf{P} \mathbf{H}^l\mathbf{P} ^ { T } \}$  can also be permuted into critical parameters $\{\boldsymbol{w},\mathbf{H}^l \}$ with permutation matrix $\mathbf{P}^T$. Thus, for any optimal solution $\boldsymbol{\mu}^\prime \in \Upsilon^\prime$, the permutation invariant function $f$ guarantees $\mathbf{P}^T\boldsymbol{\mu}^\prime \in \Upsilon$. Recall the definition of set $\hat{\Upsilon}$ and $\mathbf{P}\mathbf{P}^T = \mathbf{I}$, we then have $\boldsymbol{\mu}^\prime = \mathbf{P}(\mathbf{P}^T\boldsymbol{\mu}^\prime ) \in  \hat{\Upsilon}$. This completes the proof.

\end{document}